\documentclass{article}
\usepackage{amsmath,amssymb,epsfig,theorem,fullpage,boxedminipage}

\newtheorem{theorem}{Theorem}

\newtheorem{claim}{Claim}

\newtheorem{example}{Example}
\newcommand{\qed}{$\square$}
\newenvironment{proof}{\noindent{\em Proof:}}{\hfill \qed \medskip}
\makeatletter
\renewcommand{\section}{\@startsection{section}{1}{0pt}{-12pt}{5pt}{\large\bf}}
\renewcommand{\subsection}{\vspace*{-.1in}\@startsection{subsection}{2}{0pt}{-12pt}{-5pt}{\normalsize\bf}}
\makeatother

\title{Sponsored Search Auctions with Markovian Users}

\author{Gagan Aggarwal\thanks{Google, Inc., 1600 Amphitheatre Pkwy, Mountain View, CA, 94043.  Email: {\tt gagana@google.com}}
  \and Jon Feldman\thanks{Google, Inc., 76 Ninth Avenue, 4th Floor, New York, NY, 10011.   Email: {\tt jonfeld@google.com}}
  \and S. Muthukrishnan\thanks{Google, Inc., 76 Ninth Avenue, 4th Floor, New York, NY, 10011.   Email: {\tt muthu@google.com}}
  \and Martin P\'al\thanks{Google, Inc., 76 Ninth Avenue, 4th Floor, New York, NY, 10011.   Email: {\tt mpal@google.com}}
}

\begin{document}

\maketitle

\newcommand{\ctr}{p}
\newcommand{\pn}{\alpha}
\newcommand{\pc}{\beta}
\newcommand{\todo}[1]{{\bf **TODO: #1**}}
\newcommand{\remove}[1]{}
\newcommand{\opt}{\text{OPT}}
\newcommand{\argmax}{\arg \max}

\newcommand{\one}{\alpha}
\newcommand{\two}{\beta}
\newcommand{\three}{\gamma}
\newcommand{\four}{\delta}
\newcommand{\five}{\zeta}
\newcommand{\six}{\eta}

\newcommand{\bidders}{{\cal B}}
\newcommand{\aecpm}{a}
\newcommand{\eff}{e}

\bigskip

\begin{abstract}
Sponsored search involves running an auction among advertisers who bid
in order to have their ad shown next to search results for specific
keywords.  Currently, the most popular auction for sponsored search is
the ``Generalized Second Price'' (GSP) auction in which advertisers
are assigned to slots in the decreasing order of their {\em score},
which is defined as the product of their bid and click-through rate.
In the past few years, there has been significant research on the
game-theoretic issues that arise in an advertiser's interaction with
the mechanism as well as possible redesigns of the mechanism, but this
ranking order has remained standard.

From a search engine's perspective, the fundamental question is: what
is the best assignment of advertisers to slots? Here ``best'' could
mean ``maximizing user satisfaction,'' ``most efficient,''
``revenue-maximizing,'' ``simplest to interact with,'' or a
combination of these.  To answer this question we need to understand
the behavior of a search engine user when she sees the displayed ads,
since that defines the {\em commodity} the advertisers are bidding on,
and its value.  Most prior work has assumed that the probability of a
user clicking on an ad is independent of the other ads shown on the
page.

We propose a simple Markovian user model that does not make this
assumption. We then present an algorithm to determine the most
efficient assignment under this model, which turns out to be different
than that of GSP.  A truthful auction then follows from an application
of the Vickrey-Clarke-Groves (VCG) mechanism.  Further, we show that
our assignment has many of the desirable properties of GSP that makes
bidding intuitive. At the technical core of our result are a number of
insights about the structure of the optimal assignment.
\end{abstract}

\section{Introduction}

Targeted advertisements on search queries is an increasingly important
advertising medium, attracting large numbers of advertisers and users.
When a user poses a query, the search engine returns search results
together with advertisements that are placed into positions, usually
arranged linearly down the page, top to bottom.  On most major search
engines, the assignment of ads to positions is determined by an
auction among all advertisers who placed a bid on a keyword that
matches the query. The user might click on one or more of the ads, in
which case (in the pay-per-click model) the advertiser receiving the
click pays the search engine a price determined by the auction.

In the past few years, the sponsored search model has been highly
successful commercially, and the research community is attempting to
understand the underlying dynamics, explain the behavior of the market
and to improve the auction algorithms.  
The most common auction being run today is the {\em Generalized Second
Price} (GSP) auction: Each bidder $i$ submits a bid $b_i$ stating the
maximum amount they are willing to pay for a click, and the bidders
are placed in descending order of $b_i \ctr_i$, where $\ctr_i$ is what
is called the {\em click-through-rate} of advertiser $i$; i.e., the
probability that a user will click on the ad, given that the user
looks at it. Much of previous research on sponsored search auctions
has fixed this sort order, and focused on understanding the
implications of different pricing schemes, assuming strategic behavior
on the part of the advertisers. We now know something about GSP's
equilibrium properties~\cite{EOS,Varian,AGM}, alternative pricing that
will make it truthful~\cite{AGM}, and to some extent, impact on the
revenue in principle~\cite{EOS} and via simulations~\cite{SL}.

However, by fixing this sort order, prior work exogenizes an important
third party in sponsored search; i.e., the {\em search engine user}.
Unfortunately, there is very little guidance on this in the
literature, even though the user's behavior is the essential
ingredient that defines
{\em the commodity} the advertisers are bidding on, and its value.  We
suggest a different framework for principled understanding of
sponsored search auctions:
\begin{itemize}
\item Define a suitable probabilistic model for search engine user
behavior upon being presented the ads.
\item Once this model is fixed, ask the traditional mechanism design
questions of how do assign the ads to slots, and how to price them.
\item Analyze the given mechanism from the perspective of the bidders
(e.g., strategies) and the search engine (e.g., user satisfaction,
efficiency and revenue).
\end{itemize}

There are certain well-accepted observations about the user's
interaction with the sponsored search ads that should inform the
model: {\bf (i)} The higher the ad is on the page, the more clicks it
gets.  {\bf (ii)} The ``better'' the ad is, the more clicks it gets,
where the ``goodness'' of an ad is related to the inherent
quality of the ad, and how well it matches the user's query.  These
properties govern not only how the auction is run but also how
advertisers think about their bidding strategy (they prefer to appear
higher and get more clicks).  Thus it is important for an auction to
have what we call {\em intuitive bidding}: a higher bid translates to
a higher position and more clicks.

In this paper, we propose a natural Markov model for user clicks,
taking the above observations into account and design an algorithm to
determine an optimal assignment of ads to positions in terms of
economic efficiency. Together with VCG pricing, this gives a truthful
auction.  We further show that the optimal assignment under this model
has certain monotonicity properties that allow for intuitive
bidding. In what follows, we will describe our technical contributions
in more detail.

\paragraph{Modeling the Search Engine User}
Previous work on sponsored search has (implicitly) modeled the user
using two types of parameters: ad-specific click-through rates
$\ctr_i$ and position-specific visibility factors
$\pn_j$.
There are some intuitive user behavior models that express overall
click-through probabilities in terms of these parameters.
One possibility is ``for each position $j$ {\em independently}, the
user looks at the ad $i$ in that position with probability $\pn_j$
then clicks on the ad with probability $\ctr_i$.''  Alternatively:
``The user picks a {\em single} position according to the distribution
implied by the $\pn_j$'s, and then clicks on the ad $i$ in that
position with probability $\ctr_i$.''  Under both these models, it
follows that the probability of an ad $i$ in position $j$ receiving a
click is equal to $\ctr_i \pn_j$, which is the so-called {\em
separability} assumption~\cite{AGM}.  From separability it follows
that GSP ordering of ads will be suitable, because GSP ordering
maximizes the total advertiser value on the page.

In both these models there is no reason {\em a priori}
that the position factors $\pn_j$ should be decreasing; this is simply
imposed because it makes sense, and it is verifiable empirically. 
Also, both suggested models assume that the probability of an ad getting
clicked is independent of {\em other ads} that appear with it on the
page, an assumption made without much justification.  It is hard to
imagine that seeing an ad, perhaps followed by a click, has no effect
on the subsequent behavior of the user.  

In designing a user model, we would like to have the monotonicity of
the positions arise naturally.  Also, each ad should have parameters
dictating their effect on the user both in terms of clicking on that
ad, as well as looking at other ads.  We propose a model based on a
user who starts to scan the list of ads from the top, and makes
decisions (about whether to click, continue scanning, or give up
altogether) based on what he sees.  More specifically, we model the
user as the following Markov process: ``Begin scanning the ads from
the top down.  When position $j$ is reached, click on the ad $i$ with
probability $\ctr_i$.  Continue scanning with probability $q_i$.''  In
this model, if we try to write the click probability of an ad $i$ in
position $j$ as $\ctr_i \pn_j$, we get that $\pn_j = \Pi_{i' \in A}
q_{i'}$, where $A$ is the set of ads placed above\footnote{Throughout the
paper, we will often refer to a position or an ad being ``higher'' or
``above'' another position or ad; this means that it is earlier on the
list, and is looked at first by the user.}
 position $j$.  Thus
the ``position factor'' in the click probability decreases with
position, and does so naturally from the model.  Also note that we do
not have separability anymore, since $\alpha_j$ depends on which ads
are above position $j$. Consequently, it can be shown that GSP
assignment of ads is no longer the most efficient.

\paragraph{Auction with Markovian users}
Given this new user model, we can now ask what the best assignment is
of ads to slots.  We will study the most efficient assignment; i.e.,
the one that maximizes total advertiser value derived from user
clicks.  It turns out that the structure of this assignment is
different than that of GSP, and indeed is more sophisticated than any
simple ranking.  The presence of the $q_i$'s requires a delicate
tradeoff between the click probability of an ad and its effect on the
slots below it.  In this paper, we identify certain structural
properties of the optimal assignment and use them to find such an
optimal assignment efficiently, not only in polynomial time, but in
near-linear time.  Given this algorithm, a natural candidate for
pricing is VCG~\cite{V,C,G}, which is clearly truthful in this
setting.

\paragraph{Intuitive Bidding}
One of the reasons why GSP is successful is perhaps because bidding
strategy is intuitive: Under GSP ranking, if an advertiser bids more,
they get to a higher position, and consequently, if they bid more,
their click probability increases.  Now that we have defined a more
sophisticated assignment function, even though VCG pricing is truthful,
the auction still may not have these intuitive properties.  Our main technical
result is to show that in our model, if a mechanism uses the most
efficient assignment, indeed position and click probabilities are
monotonic in an ad's bid (with all other bids fixed), thus preserving
this important property.  While not surprising, position-monotonicity
turns out to be rather involved to prove, requiring some delicate
combinatorial arguments, and insights into the optimal substructure of
bidder assignments.

\bigskip
In summary, we approach sponsored search auctions as a three party
process by modeling the behavior of users first and then designing
suitable mechanisms to affect the game theory between the advertiser
and the search engine.  Our work sheds some light on the intricate
connection between the user models and the mechanisms; for example,
the sort order of GSP that is currently popular (sort by
$b_i\ctr_i$) is not optimal under the Markov user model we propose here.
More powerful models will be of great interest, such as making the
continuation probability $q_i$ a function of position as well, endogenizing
the actions of the user as they navigate on the landing page, etc.  We
leave it open to design truthful auctions under such extended models.

\subsection{Related Work} Sponsored search has been an active area of research in the last
several years after the early papers explored the foundational
models~\cite{EOS,AGM,Varian, LPSV}.  In general, the motivation for
the work that followed is that sponsored search in practice is much more complex
than as described by the first models.  Some papers have taken on the
effect of advertiser budgets~\cite{BCIMS, MSVV, AMT}, as well as
analyzing bidder strategy and dynamics~\cite{BCEIJM, RW, CDEGHKMS,
FMPS, WVLL, VR, LQ}.  There have also been several papers offering
extensions to GSP, or entirely new models and mechanisms~\cite{afm,
Lahaie, MNS, FMNP, mps, MS, AG}.

\subsection{Outline}  In Section~\ref{sec:model} we define our model
formally.  In Section~\ref{sec:props}, we establish several
properties of optimal assignments in this model, including our main technical result
that position and click probability will be monotone in bid and match
our intuition.  We give our algorithm for finding an optimal
assignment in Section~\ref{sec:algorithm} which gives the truthful
auction via VCG pricing, and conclude in
Section~\ref{sec:conclusions}.

\section{Markov User Click Model}
\label{sec:model}

We consider a sponsored search auction with $n$ bidders $\bidders =
\{1, \dots, n\}$ and $k$ positions.  We will also refer to ``ad $i$,''
meaning the advertisement submitted by bidder $i$.  Each bidder $i \in
\bidders$ has two parameters, $\ctr_i$ and $q_i$.  The
click-through-rate $\ctr_i$ is the probability that a user will click
on ad $i$, given that they {\em look} at it.  The continuation
probability $q_i$ is the probability that a user will look at the next
ad in a list, given that they look at ad $i$.  

Each bidder submits a bid $b_i$ to the auction, representing the
amount that they value a click.  The quantity $\ctr_i b_i$ then
represents the value of an ``impression,'' i.e., how much they value a
user looking at their ad.  This is commonly referred to as their
``ecpm.''\footnote{The acronym ecpm stands for ``expected cost per
thousand'' impressions, where M is the roman numeral for one
thousand. We will drop the factor of one thousand and refer to $\ctr_i
b_i$ as the ``ecpm.''} Throughout, we will use the notation $\eff_i = \ctr_i
b_i$ for convenience.

Given an assignment $(x_1, \dots, x_k)$ of bidders to the $k$
positions, the user looks at the first ad $x_1$, clicks on it with
probability $\ctr_{x_1}$, and then continues looking with probability
$q_{x_1}$.\footnote{The click event and the continuation event could
in principle have some correlation, and all our results will still hold.  However since we
only consider expected value, we never use this correlation
explicitly in our analysis.}
  This is repeated with the second bidder, etc., until the
last ad is reached, or some continuation test has failed.  Thus the
overall expected value of the assignment to the bidders is
$$
\eff_{x_1} + q_{x_1} (\eff_{x_2} + q_{x_2}(\eff_{x_3} + q_{x_3}(\dots q_{x_{n'-1}}(\eff_{x_n})))).
$$

The goal of the auctioneer is to compute an assignment of ads to
positions that maximizes the overall expected value.  Given this assignment,
prices can be computed using VCG~\cite{V,C,G}: for each assigned bidder
we compute the change in others' value if that bidder were to
disappear.  This assures truthful reporting of bids under a
profit-maximizing utility function.

\section{Properties of Optimal Assignments}
\label{sec:props}

We will start analyzing some basic properties of the optimal
assignment.  Our insights will allow us to give our main results
regarding monotonicity of position and click probability, as well as
an efficient algorithm for finding this assignment.  

\subsection{Adjusted ECPM}
It turns out that the quantity $\eff_i / (1-q_i)$, which we will refer to
as the ``adjusted ecpm (a-ecpm),'' plays a central role in this model.
Intuitively, this quantity is the impression value adjusted by the
negative effect this bid has on the ads below it.  We use $\aecpm_i =
\eff_i / (1-q_i)$ for convenience.  The following theorem tells us how to
assign a set of $k$ selected ads to the $k$ positions:

\begin{theorem}
In the most efficient assignment, the ads that are placed are sorted in decreasing order of
adjusted ecpm 
$
\aecpm_i = \eff_i / (1 - q_i)
$.
\label{thm:rank}
\end{theorem}

\newcommand{\effnum}{{\hat{\eff}}}

\begin{proof}
Suppose not. Then in the ranking there are two consecutive ads $i$ and $i'$ in positions $j$ and $j+1$ where
\begin{equation}
\label{eq:rank}
\frac{\eff_i}{1 - q_i} < \frac{\eff_{i'}}{1 - q_{i'}}.
\end{equation}
The contribution of positions $j \dots n$ to the efficiency of the ranking (given that position $j$ is reached) is
$$
\eff_i + q_i (\eff_{i'} + q_{i'} \effnum)
$$  
where $\effnum$ is the efficiency of positions $j+2 \dots k$ given that position $j+2$ is reached.  If $i$ and $i'$ are switched, then the contribution would change to 
$$
\eff_{i'} + q_{i'} (\eff_{i} + q_{i} \effnum),
$$  
and nothing else would change.  So since the former is the most efficient assignment, we have
$$
\eff_i + q_i (\eff_{i'} + q_{i'} \effnum)  >  \eff_{i'} + q_{i'} (\eff_{i} + q_{i} \effnum) 
$$
and so 
$$
\eff_i  (1 - q_{i'}) \geq \eff_{i'}  (1 - q_{i}).
$$
This contradicts~\eqref{eq:rank}.
\end{proof}

While this theorem tells us how to sort the ads selected, it does not
tell us {\em which} $k$ ads to select.  One is tempted to say that
choosing the top $k$ ads by a-ecpm would do the trick; however
the following example proves otherwise:

\begin{example}
\label{ex:1}
Suppose we have three bidders and two slots, and the bidders have the following parameters:

\smallskip
\begin{center}
\begin{tabular}{rrrr}
Bidder & $\eff_i$  & $q_i$ & $\aecpm_i = \eff_i / (1-q_i)$ \\ \hline
1      & \$1    & .75   & 4      \\ 
2      & \$2    & .2    & 2.5    \\ 
3      & \$0.85 & .8    & 4.25   \\ 
\end{tabular}
\end{center}
\smallskip

Let's consider some possible assignments and their efficiency.  If we
use simple ranking by ecpm $\eff_i$, we get the assignment $(2,1)$, which has
efficiency $\$2 + .2 (\$1) = \$2.20$.  If we use simple ranking by a-ecpm $a_i$ we
get the assignment $(3,1)$ with efficiency $\$0.85 + .8 (\$1) =
\$1.65$.  It turns out that the optimal assignment is $(1,2)$ with
efficiency $\$1 + .75 (\$2) = \$2.50$.  The assigned bidders are
ordered by a-ecpm in the assignment, but are not the top 2 bidders by
a-ecpm.

Now suppose we have the same set of bidders, but now we have three
slots.  The optimal assignment in this case is $(3,1,2)$; note how
bidder 3 goes from being unassigned to being assigned the first
position.
\end{example}

\subsection{Notation for Assignments}

Theorem~\ref{thm:rank} implies that the optimal assignment can be
described by the set of assigned bidders, since this set will always
be sorted by a-ecpm.
For a set $X$ of bidders, let $\eff(X)$ be the efficiency of an
assignment that places only the set $X$, regardless of the number
of bidders in $X$; i.e., if $X = \{1, 2, 3, \dots, n'\}$ and wlog is
sorted by a-ecpm, then
$$
\eff(X) = \eff_1 + q_1 (\eff_2 + q_2(\eff_3 + q_3(\dots q_{n'-1}(\eff_{n'})))).
$$  
Similarly, let $q(X)$ be the overall continuation probability of the
set $X$ of ads: $q(X) = \Pi_{i' \in X} q_{i'}$.
If $X = \emptyset$ we have $\eff(X) = 0$ and $q(X) = 1$.  

Throughout the paper, we will also make use of the notation $(A, x, B, C, y,
\dots)$ for a particular assignment, where uppercase letters denote
sets of bidders and lowercase letters denote single bidders.  This
denotes a solution where the order of the bidders is as given, and
where the order within a set is by a-ecpm (breaking ties using
a lexicographic order on bidders). So in this example, the solution
would put all the bidders of $A$ first in a-ecpm order,
followed by bidder $x$, followed by the bidders of $B$ in a-ecpm order, etc.  Note that this notation could express suboptimal
solutions (for the particular set of bidders) if the elements are not
in order of ecpm.  We use $\eff(\cdot)$ and $q(\cdot)$ to denote the
efficiency and continuation probability of such solutions.

\subsection{Bidder Dominance}

In classical sponsored search with simple ranking, a bidder $j$ can
dominate another bidder $i$ by having higher ecpm; i.e., bidder $j$
will always appear whenever $i$ does, and in a higher position.
Example~\ref{ex:1} above shows that having a higher ecpm (or a-ecpm) does not allow a bidder to dominate another bidder in our new
model.  However, we show in this section that if she has higher ecpm
{\em and} a-ecpm, then this does suffice.  This is not only
interesting in its own right, it is essential for proving deeper
structural properties in later sections.

\begin{theorem} 
\label{lemma:sub}
For all bidders $i$ in an optimal assignment, if some bidder $j$ is not in the
assignment, and $\aecpm_j \geq \aecpm_i$ and $\eff_j
\geq \eff_i$, then we may substitute $j$ for $i$, and the assignment is
no worse.
\end{theorem}

\begin{proof}
Consider some bidder $i$ that appears in an optimal assignment, and some $j$ that does not appear in the assignment, such that  $\eff_j \geq \eff_i$ and
$
\aecpm_j \geq \aecpm_i
$.
Let $(X,i,Y)$ be the optimal assignment, where $X$ is the sequence of bidders above $i$ and $Y$ is the sequence below $i$. 
The efficiency of the assignment $(X,i,Y)$ is 
$
\eff(X) + q(X) (\eff_i + q_i \eff(Y)).
$
The efficiency of the assignment $(X,j,Y)$ is 
$
\eff(X) + q(X) (\eff_j + q_j \eff(Y)).
$
Suppose $q_j \geq q_i$; then clearly $(X,j,Y)$ is as efficient as $(X,i,Y)$ since $\eff_j \geq \eff_i$, and the theorem is proven.  Thus we assume that $q_i > q_j$ for the remainder of the proof.
Note that $\aecpm_j \geq \aecpm_i$ is equivalent to:
\begin{eqnarray}
\eff_j - \eff_i & \geq & \frac{\eff_i(q_i - q_j)}{1-q_i} \label{eq:manip}
\end{eqnarray}
Now consider the assignment $(X, Y)$, with efficiency $\eff(X) + q(X) \eff(Y)$.  Since $(X, i, Y)$ is optimal, we get 
$$
\eff(X) + q(X) \eff(Y) \leq \eff(X) + q(X) (\eff_i + q_i \eff(Y)),
$$ i.e., 
$
\eff(Y) \leq \eff_i / (1 - q_i)
$.
Combining this with~\eqref{eq:manip}, and using the fact that $q_i > q_j$, we get
$
\eff_j - \eff_i \geq \eff(Y) (q_i - q_j)
$
which can be rewritten as
$$
\eff_j + q_j \eff(Y) \geq \eff_i + q_i \eff(Y).
$$
This implies that the solution $(X, j, Y)$ is as efficient as $(X, i, Y)$. 
\end{proof}

\subsection{Subset Substructure in Optimal Assignments}
In this section we give a theorem that shows some
subset structure between optimal assignments to different numbers of slots.
This theorem is used to prove position monotonicity, and is an essential ingredient of
our algorithm.  Let $\opt(C,j)$ denote the set of all optimal
solutions for filling $j$ positions with bidders from the set $C$.

\begin{theorem} 
\label{thm:subsets}
Let $j \in \{1, \dots, k\}$ be some number of positions, and let $C$
be an arbitrary set of bidders.  Then, for all $S \in \opt(C,j-1)$,
there is some $S' \in \opt(C,j)$ where $S' \supset S$.
\end{theorem}

\begin{proof}
We proceed by induction on $j$, the base case $j=1$ being simple.
Let $S$ be some solution in $\opt(C,j-1)$, and let $S'$ be the solution in $\opt(C,j)$ with the most bidders in common with $S$.
We will refer to an ad being ``above'' another ad if it has higher a-ecpm.
Let $x$ be the highest bidder in $S'$.

If $x$ does not appear in $S$, then we claim that the solution $(x,
S)$ must be in $\opt(C,j)$: Once $x$ is chosen for $S'$, taking any
set in $\opt(C \setminus x, j-1)$ for the remaining positions will
result in an optimal solution; the set $S$ is such a set, since by
assumption it does not include $x$, and is in $\opt(C, j-1)$.  But if
$(x,S) \in \opt(C,j)$ the theorem is proven, so we may assume $x \in
S$.

Let $A$ be the set of ads in $S$ above $x$, and so we can write $S =
(A, x, Q)$ and $S' = (x, Q')$.  We claim $Q' \supset Q$.  To see this,
consider the set $L$ of all ads that have lower a-ecpm than
$x$.  By the optimality of $S$, we have $Q \in \opt(L, j')$ for $j' =
j - |A| - 1 < j$. By induction there is a $Q'' \in \opt(L,j-1)$ where
$Q'' \supset Q$.  Thus we must have $Q' \supset Q$ since $S'$ is the
solution with the most bidders in common with $S$.
Decompose $Q'$ as  $Q' = (B + X, D, z, E)$ where 
\begin{itemize}
\item
$z$ is the lowest ad in $S'$ that does not appear in $S$, 
\item
$E$ is the set of ads below $z$ in $S'$ (this can be empty),
\item
$D$ is the maximal set of ads immediately above $z$ in $S'$ that also appear in $S$ (this can be empty),
\item
$X$ are the remaining ads in $S'$ that do not appear in $S$,
\item
$B$ are the remaining ads (besides $x$) that appear in both $S$ and $S'$.
\end{itemize}
Let $B' = B \cup x$.
Note that by the definitions above we may write $S = (A, B', D, E)$.
We have $\eff(S) = \eff(A, B', D, E) \geq \eff(B' + X, D, E)$ since $S \in \opt(C, j-1)$ and $|(B' + X, D, E)| = j-1$.  Decomposing this a bit gives
\begin{multline}
\label{eq1}
\eff(A, B', D) + q(A) q(B') q(D) \eff(E) \\ \geq \eff(B' + X) + q(B') q(X) \eff(D, E).
\end{multline}
We also note that
\begin{eqnarray}
\label{eq2}
\eff(S') & = & \eff(B' + X, D, z, E) \nonumber \\ & = & \eff(B' + X) + q(B') q(X) \eff(D, z, E).
\end{eqnarray}

Let $S'' = (A, B', D, z, E)$.  The remainder of the proof will show
that $\eff(S'') \geq \eff(S')$, which implies the theorem since $|S''|
= j$, $S'' \supset S$ and $S' \in \opt(C, j)$.  We can rewrite
$\eff(S'')$ as follows:
\begin{eqnarray}
\eff(S'') & = & \eff(A, B', D) + q(A) q(B') q(D) \eff(z, E) \nonumber \\
& \geq & \eff(B' + X) + q(B') q(X) \eff(D, E)  \nonumber \\
& & +~q(A) q(B') q(D) (\eff(z,E) - \eff(E))  \label{eq:sppline1} \\
& = & \eff(S') - q(B') q(X) (\eff(D, E) - \eff(D, z, E)) \nonumber \\ 
&& +~q(A) q(B') q(D) (\eff(z,E) - \eff(E)) \label{eq:sppline2} 
\end{eqnarray}
In the above,~\eqref{eq:sppline1} follows from~\eqref{eq1},
and~\eqref{eq:sppline2} follows from~\eqref{eq2}.
  Rearranging, and
using $$\eff(D, E) - \eff(D, z, E) = q(D)(\eff(E) - \eff(z, E)),$$ we get
$$
\eff(S'') - \eff(S') \!=\! q(B') q(D) (q(A) - q(X))(\eff(z, E) - \eff(E)).
$$

We know that $\eff(z, E) \geq \eff(E)$ since otherwise $\eff(B' + X,
D, E) > \eff(B' + X, D, z, E) = \eff(S')$, and this cannot be since
$S' \in \opt(C, j)$.  We claim that $q(A) \geq q(X)$, which would
imply $\eff(S'') \geq \eff(S')$ and thus complete the proof.  
This
 is trivially true if $A = X = \emptyset$.  Since $|A| = |X|$ by the definitions above, 
we can assume both $A$ and $X$ are non-empty.
consider some $y
\in A$ and ${y'} \in X$.  We have $y \notin S'$ by the definition of $A$.  Since $S'$ is the solution in $\opt(C, j)$
with the most bidders in common with $S'$, we must not be able to
substitute $y$ for ${y'}$ in $S$, and thus by Lemma~\ref{lemma:sub}
we must have that $\eff_{y'} > \eff_y$ or $\aecpm_{y'} > \aecpm_{y}$.  But by the definitions of $A$ and $X$, we have
$\aecpm_y \geq \aecpm_x \geq \aecpm_{y'}$.  Therefore $\eff_{y'} >
\eff_y$.  The previous two inequalities imply $q_y > q_{y'}$.  Since
$y$ and ${y'}$ were arbitrary and $|A| = |X|$, this gives $q(A) > q(X)$.
\end{proof}

\subsection{Monotonicity of Position and Click Probability}

In this section we give our main theorem regarding the structure of
the optimal assignments in the Markovian click model: that position
and click probability are monotonic in a bidder's bid, with all other
bids fixed.  This is a fundamental property that makes the bidder's
interaction with the system intuitive, and allows the bidder to
adjust her bid intelligently without global knowledge of the other
bids.

\begin{theorem}
\label{thm:monotonic}
With all other bids fixed, the probability of receiving a click in the
optimal solution is non-decreasing in one's bid.  In addition, the
position of a particular bidder in the optimal solution is monotonic
that bidder's bid.
\end{theorem}

\begin{proof}
As bidder $x$ increases her bid $b_x$ with all other bids fixed, the
value of a particular solution $S = (A,x,B)$ increases linearly as
$q(A) \ctr_x \cdot b_x + [\eff(A) + q(A) q_x \eff(B)]$.  (Solutions
not involving $x$ stay constant.)

Let $S_1, \dots, S_k$ denote the sequence of optimal solutions that
occurs as $b_x$ increases from $0$.  Solution $S_1$ is the best
solution not involving $x$, and $S_k$ is the best solution that puts
$x$ in the first position.  By the fact that each solution increases
linearly by the term $q(A) \ctr_x \cdot b_x$, which is the probability that $x$
receives a click in that solution, it must be the case that for a new
solution to become optimal it gives $x$ a higher click probability
than in the previous solution; i.e., for all $i \geq 0$, $S_{i+1}$
gives $x$ a higher click probability than $S_{i}$.  This proves the
first part of the theorem.

Now suppose the second part of the theorem is false.  Then, there must
be some consecutive solutions $S_i$ and $S_{i+1}$ where $x$ has a
higher position in $S_i$ than in $S_{i+1}$.  Let $b$ be the bid that
makes both $S_i$ and $S_{i+1}$ optimal, which must exist since they
are consecutive in the list of optimal solutions, and fix $b$ for the
remainder of the proof.  Decompose the two solutions as $S_{i+1} = (A,
x, E)$ and $S_i = (F, x, G)$ where $|A| > |F|$ by assumption, and
$q(A) > q(F)$ by the argument that proved the first part of the
theorem.  Since both solutions are optimal, they are both sorted by
a-ecpm, with ties broken lexicographically.

We claim that $A \cap G = \emptyset$. If this were not the case then
some bidder $y$ would appear in both $A$ and $G$, but since both
$S_i$ and $S_{i+1}$ are sorted by ecpm, and $y$ appears on different
sides of $x$, this must mean that $y$ has the same a-ecpm as $x$.
But, this violates our assumption on how the algorithm breaks ties
among different orderings of the same set.  Using similar logic, we
get $F \cap E = \emptyset$.

By the optimality of $S_{i+1}$ and the fact that $F \cap E =
\emptyset$, we get $ E \in \opt({\cal B} \setminus (F \cup A \cup
\{x\}), n - |A| - 1) $.  Since $|F|<|A|$, Theorem~\ref{thm:subsets}
then implies that there is some $ G' \in \opt({\cal B} \setminus (F
\cup A \cup \{x\}), n - |F| - 1) $ where $G' \supset E$.  Even if $G'
\neq G$, the set $G'$ could replace $G$ in $S_i$ (since $A \cap G = \emptyset$) and still be optimal,
and so we define $S'_i = (F, x, G')$ and have $\eff(S'_i) =
\eff(S_i)$.  Note that by the definition of $G'$, we have $A \cap G' = \emptyset$.

Since $A \cap G' = \emptyset$ and $|A| > |F|$, there must be some
bidder in $A$ that does not appear in $S'_i$.  Let $a$ be the first
such bidder (by a-ecpm).  Decompose $A$ into $(C, a, D)$ where $C$ and
$D$ are those bidders with higher and lower a-ecpm than $a$,
respectively.  Note that $C \subseteq F$, by the definition of
$a$. Let $F_1$ be the smallest prefix of $F$ (by a-ecpm) that contains
all of $C$, and let $F_2 = F - F_1$.

Similarly, since $A \cap G' = \emptyset$ and $|G| > |E|$, we let $t \in
G'$ be the first bidder (by a-ecpm) that does not appear in $S_{i+1}$.  Let
$G'_1$ be the bidders in $G'$ with higher a-ecpm than $t$, and $G'_2 = G'
- G'_1 - t$.  Note $G'_1$ is also a prefix of $E$, and let $E_2 = E - G'_1$.

Given these definitions, we define nine different solutions that we
will use in our proof (renaming $S'_i$ and $S_{i+1}$ for clarity):

\begin{eqnarray*}
\one &=& S_{i+1} = (A, x, E) = (C, a, D, x, E)\\
\one' &=& (C, D, x, E) = (C, D, x, G'_1, E_2)\\
\two &=& S'_i = (F, x, G') = (F, x, G'_1, t, G'_2)\\
\two' &=& (F, x, G'_1, G'_2)\\
\three &=& (C, D, x, E+t) = (C, D, x, G'_1, t, E_2) \\
\four &=& (F_1, a, F_2, x, G'_1, G'_2)\\
\five &=& (C, a, F_1 - C, F_2, x, G'_1, G'_2)\\
\five' &=& (C, F_1 - C, F_2, x, G'_1, G'_2)\\
\six &=& (C, F_1 - C, a, F_2, x, G'_1, G'_2)
\end{eqnarray*}

In the following claims we will often use the generic identity 
\begin{equation}
\label{eq:identity}
\eff(X, y, Y) - \eff(X, Y) = q(X)(\eff_y - (1-q_y)\eff(Y))
\end{equation}

\begin{claim}
\label{claim:scomp1}
$\eff(\three) - \eff(\one') > \eff(\two) - \eff(\two')$
\end{claim}

\begin{proof}
Using~\eqref{eq:identity}, we can rewrite the claim as
\begin{multline}
q(C,D,x,G'_1)(\eff_t - (1 - q_t)\eff(E_2)) \\
> 
q(F,x,G'_1)(\eff_t - (1 - q_t)\eff(G'_2)).
\end{multline}
Since $q(C,D,x) \geq q(C,a,D,x) = q(A,x)$, and $q(F) < q(A)$, we get $q(C,D,x,G'_1) > q(F,x,G'_1)$, and so it remains to
prove $\eff(E_2) \leq \eff(G'_2)$.  But this follows from the optimality of
$\two$, since $E_2$ could replace $G'_2$ in solution $\two$ (indeed, $E_2
\subset G'_2$).
\end{proof}

\begin{claim}
\label{claim:scomp2}
$\eff(\five) - \eff(\five') \geq \eff(\one) - \eff(\one')$
\end{claim}
\begin{proof}
Using~\eqref{eq:identity}, we can rewrite the claim as
\begin{multline}
q(C)(\eff_a - (1-q_a)) \eff(F_1 - C, F_2, x, G'_1, G'_2)) \\ \geq
q(C)(\eff_a - (1-q_a) \eff(D,x,E)).
\end{multline}
Since $(F_1 - C, F_2, x, G'_1, G'_2)$ has the same length as
$(D,x,E)$, and does not contain $a$ or any bidders in $C$, it could
replace $(D,x,E)$ in $\one$; but since $\one$ is optimal, we may
conclude that $\eff(D,x,E) \geq \eff(F_1 - C, F_2, x, G'_1, G'_2)$, which proves the claim.
\end{proof}

Note that $\eff(\six) \geq \eff(\five)$, since all bidders in $F_1$ have a
higher a-ecpm than $a$.  Finally, note that a simple application
of~\eqref{eq:identity} gives
\begin{equation}
\label{eq:scomp3}
\eff(\six) - \eff(\five') = \eff(\four) - \eff(\two')
\end{equation}
We now conclude the proof with the following contradiction:

\medskip

\begin{tabular}{rcll}
$\eff(\one) - \eff(\one')$ 
& \!$\geq$\! & $\eff(\three) - \eff(\one')$ & (by opt. of $\one$) \\ 
& \!$>$\!    & $\eff(\two) - \eff(\two')$ & (Claim~\ref{claim:scomp1}) \\
& \!$\geq$\! & $\eff(\four) - \eff(\two')$ & (by opt. of $\two$) \\
& \!$=$\!    & $\eff(\six) - \eff(\five')$ & (by~\eqref{eq:scomp3}) \\
& \!$\geq$\! & $\eff(\five) - \eff(\five')$ &  \\
& \!$\geq$\! & $\eff(\one) - \eff(\one')$ & (by Claim~\ref{claim:scomp2})
\end{tabular}

\end{proof}

\section{Computing the Optimal Assignment}
\label{sec:algorithm}

In this section we give algorithms for computing the optimal
assignment of bidders to positions using the structural properties
we proved in the previous section. We begin with a simple dynamic
program that gives an $O(n \log n + nk)$ time algorithm.  We then show how our
insights from the previous sections give a faster $O(n\log n + k^2
\log^2 n)$ time algorithm.

\subsection{Optimal Assignment using Dynamic Programming}
The algorithm proceeds as follows.  First, sort the ads in decreasing
order of a-ecpm in time $O(n \log n)$.  Then, let $F(i,j)$ be
the efficiency obtained (given that you reach slot $j$) by filling
slots $(j, \dots, k)$ with bidders from the set $\{i, \dots, n\}$.  We
get the following recurrence:
$$
F(i,j) = \max (F(i+1, j+1) q_i + \eff_i, F(i+1, j)).
$$ 
Solving this recurrence for $F(1,1)$ yields the optimal assignment,
and can be done in $O(nk)$ time.

\subsection{Near-linear Time Algorithm}
Let ${\cal B} = \{1, \dots, n\}$ be the set of bidders, sorted by
a-ecpm.  Suppose we had an oracle that told us, for any $j,j' \in
{\cal B}$, the bidder $y$ with $j \leq y \leq j'$ that maximizes
$f(q_y, \eff_y)$ for an arbitrary linear function $f$.  We will later
show how to construct this oracle, but first we describe our algorithm
that uses this oracle.

Our algorithm will construct a solution $S_i \in \opt({\cal B}, i)$
for all $i = 1, \dots, k$, the final one $S_k$ being the overall
optimum.  By Theorem~\ref{thm:subsets}, we may assume that $S_{i+1}
\supset S_i$.  Using this fact, our algorithm builds $S_{i+1}$ from
$S_i$ by simply finding $\argmax_{x \notin S_i} \eff(S_i \cup \{x\})$.
To perform this max, the algorithm first guesses (i.e., searches
exhaustively for) the a-ecpm rank of the new bidder $x$ among the bidders
in $S_i$; this a number $\ell$ from $1$ to $i+1$.
Let $\{s_1, \dots, s_i\}$ be the elements of $S_i$ sorted by
decreasing a-ecpm.  The new bidder $x$ has a-ecpm between $s_{\ell-1}$ and $s_{\ell}$, and so
$
\eff(S_i \cup \{x\}) = \eff(s_1, \dots, s_{\ell-1}) + q(s_1, \dots, s_{\ell-1}) (\eff_x + q_x \eff(s_{\ell}, \dots, s_i))
$.
Since $\eff(S_i \cup \{x\})$ is linear in $(q_x, \eff_x)$, we 
may appeal to the oracle to find
the bidder $x$ that maximizes $\eff(S_i \cup \{x\})$ among all bidders
with a-ecpm between that of $s_{\ell-1}$ and $s_\ell$.  We make $i+1$ calls to
this oracle for each $i$, and thus $O(k^2)$ calls overall.  To get the
coefficients of $\eff(S_i \cup \{x\})$ to pass to the oracle, we
precompute the quantities $q(s_1, \dots, s_p)$ and $\eff(s_p, \dots,
s_i)$ for all $p$.  (This can be done in $O(k)$ time per $i$,
for $O(k^2)$ time overall.)

It remains to show how to implement the oracle.  We first preprocess
the sequence $[1, \dots, n]$ of bidders as follows.  We consider the
{\em dyadic} intervals $[\alpha 2^\beta + 1, \dots,
(\alpha+1)2^\beta]$ for each possible $\alpha, \beta$, for a total of
$O(n)$ intervals.  Note that any subsequence $[j,\dots,j']$ is made up
of at most $O(\log n)$ such intervals.  For each such interval, we
will make a data structure that can find $\max f(q_x, \eff_x)$ over
bidders in that interval in $O(\log n)$ time.  So overall, given $[j,
\dots, j']$, the oracle takes the max of $O(\log n)$ calls to the data
structure, and completes in $O(\log^2 n)$ time.

The data structure we compute for a particular interval $[\alpha
2^\beta + 1, \dots, (\alpha+1)2^\beta]$ is simply the convex hull of
the points $(q_x, \eff_x)$ in two-dimensional space defined by bidders
$x$ in the interval.  We can compute all these convex hulls in $O(n
\log n)$ time by successively merging convex hulls for increasing
$\beta$.  Given the convex hull (with the points sorted in order of
$q_x$ for example), a simple binary search can find the point
maximizing $f(q_x, \eff_x)$ in $O(\log n)$ time.

This gives,

\begin{theorem}
Consider the auction with $n$ Markovian bidders and $k$ slots. 
There is an optimal assignment which can be determined in $O(n\log n + k^2 \log^2 n)$ time.
\end{theorem}

It follows that using VCG pricing with this optimal assignment, we obtain a truthful
mechanism for sponsored search with Markovian users. 

\section{Concluding Remarks}
\label{sec:conclusions}

We approached sponsored search auctions as a three party process by
modeling the behavior of users first and then designing suitable
mechanisms to affect the game theory between the advertiser and the
search engine.  This formal approach shows an intricate connection between 
the user models and the mechanisms.

There are some interesting open issues to understand about our model
and mechanism.  For example, in order to implement our mechanism, the
search engine needs to devise methods to estimate the parameters of
our model, in particular, $q_i$'s.  This is a challenging statistical
and machine learning problem.  Also, we could ask how much improvement
in efficiency and/or revenue is gained by using our model as opposed
to VCG without using our model. 

More powerful models will also be of great interest.  One small
extension of our model is to make the continuation probability $q_i$ a
function of location as well, which makes the optimization problem
more difficult.  We can also generalize the Markov model to handle
arbitrary configurations of ads on a web page (not necessarily a
search results page), or to allow various other user states (such as
navigating a landing page).  Finally, since page layout can be
performed dynamically, we could ask what would happen if the layout of
a web page were a part of the mechanism; i.e., a function of the bids.

\bibliographystyle{plain}
\bibliography{sponsored_search}

\end{document}